\documentclass[12pt,formal]{article}

\usepackage{etex}
\usepackage{algcompatible}
\usepackage{algpseudocode,algorithm,algorithmicx}
\usepackage{fullwidth}
\usepackage{etoolbox}
\usepackage{ifthen}
\usepackage{float}
\usepackage{tikz}
\usepackage{multicol}
\usepackage{todonotes}
\usepackage{wrapfig}
\usepackage{adjustbox}
\usepackage{tkz-graph}
\usetikzlibrary{arrows,automata}
\usetikzlibrary{arrows,shapes,positioning,scopes}
\usetikzlibrary{tikzmark,decorations.pathreplacing,calc,fit}
\usetikzlibrary{arrows,automata}
\usetikzlibrary{arrows,petri}
\GraphInit[vstyle = Shade]
\usepackage{ulem}
\usepackage[latin1]{inputenc}
\usepackage{ifthen}
\usepackage{mathrsfs}
\usepackage{multirow}
\usepackage{rotating}
\usetikzlibrary{matrix,arrows}
\usepackage{bussproofs}
\usepackage{qtree}
\usepackage{easyReview}
\usepackage[polutonikogreek,english]{babel}
\usepackage{tikz}
\usepackage{adjustbox}

\usetikzlibrary{arrows,petri}
\usetikzlibrary{arrows.meta}

\usepackage{pgf}
\usepackage{times}
\usepackage[T1]{fontenc}
\usepackage{amsfonts,amsmath,amssymb,amsthm}
\usepackage[all]{xy}

\usepackage{dingbat}
\usepackage{subcaption}
\usepackage{cmll}
\floatstyle{plain}
\floatstyle{ruled}
\restylefloat{algorithm}
\newfloat{myalgo}{tbhp}{mya}
\newcommand{\INDSTATE}[1][1]{\State\hspace{1cm}}
\newlength\myindent
\algrenewcommand\algorithmicrequire{\textbf{Precondition:}}
\algrenewcommand\algorithmicindent{1mm}

\newcommand{\size}[1]{\ensuremath{\left|#1\right|}}

\newcommand{\imply}{\supset}

\newcommand{\mil}{\ensuremath{\mathbf{M}_{\imply}}}

\newcommand{\SP}{$\mathcal{S}\mathcal{P}$}

\newcommand{\Branch}[1]{\overrightarrow{#1}}

\newboolean{numberspec}
\setboolean{numberspec}{false}
\newcounter{specline}

\newcommand{\spec}[3][{}]{%
   \setcounter{specline}{0}%
   \ifthenelse{\equal{#1}{numbers}}%
      {\setboolean{numberspec}{true}}%
      {\setboolean{numberspec}{false}}%
   \ensuremath{
      \begin{array}{l}
         \ifthenelse{\equal{#2}{}}
            {\numberedline}
            {#2\nl}
         #3
      \end{array}
   }
}

{\begin{myalgo}[#1]
\centering
\begin{minipage}{#2}
\begin{algorithm}[H]}%
{\end{algorithm}
\end{minipage}
\end{myalgo}}

\newbool{formal}
\booltrue{formal}

\input{tcilatex}

\begin{document}
  \title{On the Intrinsic Redundancy in Huge Natural Deduction proofs II: Analysing $M_{\imply}$ Super-Polynomial Proofs}
  \author{Edward Hermann Haeusler$^{1}$ \\ Departmento de Inform\'{a}tica\\ PUC-Rio\\
    Rio de Janeiro, Brasil \\
   $^1$Email:hermann@inf.puc-rio.br}

\maketitle
  
   \begin{abstract}

     This article precisely defines huge proofs within the system of Natural Deduction for the Minimal implicational propositional logic \mil. This is what we call an unlimited family of super-polynomial proofs. We consider huge families of expanded normal form mapped proofs, a device to explicitly help to count the E-parts of a normal proof in an adequate way. Thus, we show that for almost all members of a super-polynomial family there at least one sub-proof or derivation of each of them that is repeated super-polynomially many times. This last property we call super-polynomial redundancy. Almost all, precisely means that there is a size of the conclusion of proofs that every proof with conclusion bigger than this size and that is huge is highly redundant too. This result points out to a refinement of compression methods previously presented and an alternative and simpler proof that CoNP=NP.

\end{abstract}


\section{Introduction}

This article describes a result that is an extension of the obtained in \cite{Exponential}. Both articles discuss the relationship between the size of proof and how redundant it is. The redundancy in a proof or logical derivation is related to the fact that this proof has sub-proofs that are repeated many times inside it. The focus of this article is Natural Deduction (ND) proofs in the purely implicational minimal logic $\mil$. The reason to work with this logic lies in the fact that $\mil$ is PSPACE-complete and simulates polynomially any proof in Intuitionistic Logic and full minimal logic, being hence an adequate representative to study questions regarding computational complexity. The fact that $\mil$ has a straightforward syntax and an ND system too is worthy of notice. Moreover, compressing proofs in $\mil$ can provide very good glues to compress proofs in any one of these mentioned systems, even for the Classical Propositional Logic. One of the reasons to study redundancy is to obtain a compressing method based on redundancies removals. 
In \cite{Exponential} we identify sets of huge proofs with sets of proofs that, when viewed as strings, have their length lower-bounded by some exponential function. Moreover, we can consider, without loss of generality,  proof/deductions, which are square height bounded, as stated in \cite{Studia2019}.  In \cite{Exponential}, we prove that the exponentially lower-bounded $\mil$ proofs are redundant, in the sense that there is at least one sub-proof for each proof that occurs exponentially many times in it.

Here, in this article,  we show that this result extends to super-polynomial proofs, i.e., proofs that are lower-bounded by any polynomial. In this article, we go further and identify huge proofs/derivations with super-polynomial sized proofs.  We prove that, in any set of super-polynomially lower-bounded proofs in $\mil$, of some tautologies, almost all proof is redundant. Redundancy means that for almost every proof in this set of huge proofs, there is a sub-proof that occurs super-polynomially many times in it. The technique used in this new proof has a structure that is quite similar to the proof reported in \cite{Exponential}. In order to facilitate the task of reading this article, without deviating the reader to read the material in the previous work, we will repeat here the main definitions and the technical part in \cite{Exponential}, contextualizing for this article.  

In section~\ref{sec:Background}, we present the background and terminology used. Section~\ref{sec:Main} is the section where we prove the main result of this article. In conclusion, section~\ref{sec:Conclusion}  we briefly discuss the use of the result proven here to obtain theoretical compressions methods that provide a super-polynomial compression ratio, for some super-polynomial sets of proofs,  all of them compress to polynomial size.   
Finally, we want to mention that, in \cite{Exponential}, we show three sets of exponentially lower-bounded sized proofs that are linearly bounded on the height. The reader can find these examples useful material for analyzing good concrete cases, including certificates for non-hamiltonian graphs. We want to comment, that, we do not know any concrete, and, easy to define, set of lower-bounded super-polynomial proofs that are not themselves exponentially lower-bounded, as are the examples in \cite{Exponential} and \cite{exponential}. 
Finally, in \cite{Exponential}, we prove the main results using the language of graphs and trees,  with the sake of more comfortable explanations. In this work, this is not possible any more; the use of proofs, formulas, syntax-tree and concepts from proof-theory is essential for presenting the results we deal within this article. The tree and graph terminology, however, is used. We present this in the next section.

\section{A brief explanation of Natural Deduction and some basic proof-theoretical concepts}
\label{sec:Background}

The Natural Deduction system defined by Gentzen in \cite{Gentzen1936}),  is defined as a set of rules that settle the concept of a logical deduction.
Language and inference rules can be viewed as a {\em logical calculus}, as defined by Church (\cite{Church}). In contrast with the main formulations of logical calculus for some logics by Hilbert (\cite{HilbertHis}), Natural Deduction does not have axioms. Natural Deduction implements in the level of the logical calculus the  {\em (meta)theorem of  deduction}, namely from $\Gamma, A\vdash A\imply B$, employing the discharging mechanism. The $\imply$-introduction rule shows how this discharging mechanism implements in the logic calculus the {\em deduction theorem}.

\begin{prooftree}
  \AxiomC{[$A$]}
  \noLine
  \UnaryInfC{$\Pi$}
  \noLine
  \UnaryInfC{$B$}
  \RightLabel{$\imply$-Intro}
  \UnaryInfC{$A\imply B$}
\end{prooftree}

We embrace formulas occurrence  $A$ in the derivation $\Pi$ of $B$ from $A$ with a pair of [] to indicate the discharge of them. Embracing a formula occurrence means that from the application of the $\imply$-Intro rule discharging this occurrence of the $\imply$-Intro down to the conclusion of the derivation, the inferred formulas do not depend anymore on discharged $A$.
The choice of formulas to be discharged in an application of a $\imply$-Intro discharges is arbitrary and liberal. The range of this choice goes from every occurrence of $A$ until none of them. The following derivations show two different ways of deriving $A\imply (A\imply A)$. Observe that in both deductions or derivations, we use numbers to indicate which is the application of the $\imply$-Intro that discharged the marked formula occurrence. For example, in the right derivation, the upper application discharged the marked occurrences of $A$, while in the left derivation, it is the lowest application that discharges the formula occurrences $A$. There is a third derivation that both applications do note discharge any $A$, and the conclusion $A\imply(A\imply A)$ keep depending on $A$. This third alternative appears in figure~\ref{third}. Natural Deduction systems can provide logical calculi without any need to use axioms. In this article, we focus on the system formed only by the $\imply$-Intro rule and the $\imply$-Elim rule, as shown below, also known by {\em modus ponens}. The logic behind this logical calculus is the purely minimal implicational logic, $M_{\imply}$.

\begin{prooftree}
  \AxiomC{$A$}
  \AxiomC{$A\imply B$}
  \RightLabel{$\imply$-Elim}
  \BinaryInfC{$B$}
\end{prooftree}

 We can substitute the liberal discharging mechanism by a greedy one that discharges every possible formula occurrence whenever the $\imply$-Intro is applied. Observe that, in this case, the derivation in figure~\ref{third} would not be possible anymore. Completeness regarding derivability would be lost. However, when considering proofs, i.e.,  derivations with no assumption undischarged, the greedy version of the $\imply$-Intro is enough to ensure the demonstrability of valid formulas, see~\ref{appendix:greedy-imply-intro}, and hence for the representation of proofs we can use $\mil^{\rightarrow}$

\begin{prooftree}
  \AxiomC{$[A]^1$}
  \UnaryInfC{$A\imply A$}
  \RightLabel{$\;^1$}
  \UnaryInfC{$A\imply (A\imply A)$}
  \AxiomC{$[A]^1$}
  \RightLabel{$\;^1$}
  \UnaryInfC{$A\imply A$}
  \UnaryInfC{$A\imply (A\imply A)$}
  \noLine
  \BinaryInfC{}
\end{prooftree}

\begin{figure}[H]
\begin{prooftree}
  \AxiomC{$A$}
  \UnaryInfC{$A\imply A$}
  \UnaryInfC{$A\imply (A\imply A)$}
\end{prooftree}
\caption{Two vacuous $\imply$-Intro applications}\label{third}
\end{figure}

In \cite{Exponential}, with the sake of having simpler proofs of our results, we consider Natural Deduction as trees. For any derivation in ND, there is a binary tree having nodes labelled by the formulas and edges linking premises to conclusion, such that the root of the tree would be the conclusion of the derivation, and the leaves are its assumptions. The derivation in figure~\ref{ex:derivacao} has the tree in figure~\ref{ex:tree} representing it.  The set of formulas that the label of $u$ depends on the label of $v$  labels the edge from  $v$ to  $u$. This set of formulas is called the dependency set of the label of $u$ from the label of $v$. In this way, the $\imply$-intro, in fact, its greedy version,  removes the discharged formula from the dependency set, as shown in figure~\ref{ex:tree}. Note that because of this labelling of edges by dependency sets, we need one more extra edge and the root node. The dependency set of the conclusion labels this edge. That is the reason for the edge linking the conclusion to the dot in figure~\ref{ex:tree}.

\begin{figure}[H]
  \begin{prooftree}
  \AxiomC{$[A]^{1}$}
  \AxiomC{$A\imply B$}
  \BinaryInfC{$B$}
  \AxiomC{$B\imply C$}
  \BinaryInfC{$C$}
  \UnaryInfC{$\LeftLabel{1}A\imply C$}
  \end{prooftree}
\caption{A derivation in $M_{\imply}$}\label{ex:derivacao}
\end{figure}

\begin{figure}[H]
      \begin{tikzpicture}[shorten >=1pt,node distance=1cm,auto]
\tikzstyle{state}=[shape=circle,thick,minimum size=1.5cm]
\node[state] (ground) {.};
\node[state, above of=ground] (AC) {$A\imply C$};
\node[state, above of=AC] (C) {$C$};
\node[state, above left of=C] (B) {$B$};
\node[state, above left of=B] (A) {$A$};
\node[state, above right of=B] (AB) {$A\imply B$};
\node[state, above right of=C] (BC) {$B\imply C$};
\path [draw,->] (ground) edge node[midway] {{\tiny $\{A\imply B, B\imply C\}$}} (AC);
\path [draw,->] (AC) edge node[midway] {{\tiny $\{A,A\imply B, B\imply C\}$}} (C);
\path [draw,->] (C) edge node[midway,right] {{\tiny $\{B\imply C\}$}} (BC);
\path [draw,->] (C) edge node[midway, left] {{\tiny $\{A,A\imply B\}$}} (B);
\path [draw,->] (B) edge node[midway, right] {{\tiny $\{A\imply B\}$}} (AB);
\path [draw,->] (B) edge node[midway, left] {{\tiny $\{A\}$}} (A);
    \end{tikzpicture}
\caption{The tree representing the derivation in figure~\ref{ex:derivacao}}\label{ex:tree}
\end{figure}
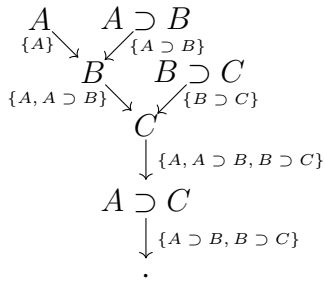

We use bitstrings induced by an arbitrary linear ordering of formulas in order to have a more compact representation of the dependency sets. Taking into account that only subformulas of the conclusion can be in any dependency set, we only need bitstrings of the size of the conclusion of the proof. In figure~\ref{ex:bitstring} we show this final form of tree representing the derivation in figure~\ref{ex:derivacao} and~\ref{ex:tree}, when the linear order $\prec$ is $A\prec B\prec C\prec A\imply B\prec B\imply C\prec A\imply C$. 

\begin{figure}[H]
    \begin{tikzpicture}[shorten >=1pt,node distance=1cm,auto]
\tikzstyle{state}=[shape=circle,thick,minimum size=1.5cm]
\node[state] (ground) {.};
\node[state, above of=ground] (AC) {$A\imply C$};
\node[state, above of=AC] (C) {$C$};
\node[state, above left of=C] (B) {$B$};
\node[state, above left of=B] (A) {$A$};
\node[state, above right of=B] (AB) {$A\imply B$};
\node[state, above right of=C] (BC) {$B\imply C$};
\path [draw] (ground) edge node[midway] {{\tiny $000110$}} (AC);
\path [draw] (AC) edge node[midway] {{\tiny $100110$}} (C);
\path [draw] (C) edge node[midway,right] {{\tiny $000010$}} (BC);
\path [draw] (C) edge node[midway, left] {{\tiny $100100$}} (B);
\path [draw] (B) edge node[midway, right] {{\tiny $000100$}} (AB);
\path [draw] (B) edge node[midway, left] {{\tiny $100000$}} (A);
    \end{tikzpicture}
\caption{Tree with bitstrings representing the derivation in figure~\ref{ex:derivacao}}\label{ex:bitstring}
\end{figure}
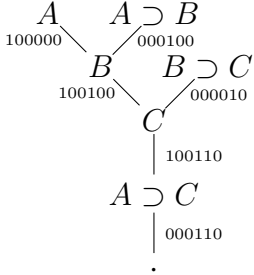



In the sequel of this section, we only briefly present the list of main results and definitions, from Natural Deduction proof-theory,  easing the reader's task to understand our proof. 

We explain the primary rationale of this article as follows. The sub-formula principle for a logic $\mathcal{L}$ states that all we need to prove a tautology is inside itself. That is, without loss of generality (w.l.g.), If $\alpha$ is a tautology, then there is a proof of $\alpha$ using only sub-formulas of $\alpha$ in it. This property is a corollary of the Normalization theorem for Natural Deduction, a central result, and a tool of proof-theory. Well, \mil satisfies the normalization and hence the sub-formula principle. We note that the amount of sub-formulas of any formula is linear on its size. A Natural Deduction proof is huge whenever its size is larger than or equal to any polynomial on the size of its conclusion\footnote{If we follow Cook-Karp conjecture that says that computationally easy to verify and to compute objects are of polynomial-size, huge proofs are the hard proofs for verification, namely, the super-polynomial ones}. Thus, if the size of a proof is bigger than any polynomial, then its corresponding labelled tree is also bigger than any polynomial.  We remind that each sub-formula is a possible label node in the tree. We have then that a  super-polynomial normal proof has to be labelled with linearly many labels, regarding the size/length of the string that labels its root. This configuration allows us to say that at least one label repeats super-polynomially many times in the tree under the additional condition that it is also linearly-height bounded. We show that this repetition induces, in some way, sub-proof repetitions, such that, this sub-proof repeats super-polynomially many times too.

Without loss of generality, we consider the additional hypothesis on the linear bound on height of the proof of \mil tautologies. In \cite{Studia2019}, we show that any tautology in \mil has a Natural Deduction normal proof of height bound by the square of the size of this tautology.
However, if we consider the complexity class $CoNP$ (see the appendix in ~\cite{Exponential}) we are naturally limited to linearly height-bounded proofs. The proofs, in \mil,  of the non-hamiltonianicity of  graphs, are linearly height bounded.



\section{Terminology and definitions}\label{terminology}

Following the usual terminology in Natural Deduction and proof-theory, we briefly describe what we use in this article. This section is strongly base in \cite{Exponential}. We consider the usual definition of the syntax tree
for \mil-formulas. Given a formula $\phi_1\mil \phi_2$ in \mil, we call $\phi_2$ its right-child and $\phi_1$ its left-child. These formulas label the respective right and left child vertexes that are labelled with them. A right-ancestral of a vertex $v$ in a syntax-tree $T_{\alpha}$ of a formula $\alpha$ is any vertex $u$, such that, either $v$ is the right-child of $u$, or, there is a vertex $w$, such that $v$ is the right-child of $w$ and $u$ is right-ancestral of $w$.  

The left premise of a $\imply$-Elim rule is called a minor premise, and the right premise is called the major premise. We should note that the conclusion of this rule, as well as its minor premise, are sub-formulas of its major premise. We also observe that the premise of the $\imply$-Intro is the sub-formula of its conclusion. A derivation is a tree-like structure built using $\imply$-Intro and $\imply$-Elim rules. We have some examples depicted in the last section. The conclusion of the derivation is the root of this tree-like structure, and the leaves are what we call top-formulas. A proof is a derivation that has every top-formula discharged by a $\imply$-Intro application in it. The top-formulas are also called assumptions. An assumption that it is not discharged by any $\imply$-Intro rule in a derivation is called an open assumption. If $\Pi$ is a derivation with conclusion $\alpha$ and $\delta_1,\ldots,\delta_n$ as all of its open assumptions then we say that $\Pi$ is a derivation of $\alpha$ from $\delta_1,\ldots,\delta_n$.

\begin{definition}[Branch]\label{def:Branch} A branch in a derivation or proof $\Pi$ is any sequence $\beta_1,\ldots,\beta_k$ of formula occurrences in $\Pi$, such that:
\begin{itemize}
\item $\delta_1$ is a top-formula, and;
\item For every $i=1,k-1$,  either $\beta_i$ is a $\imply$-Elim major premise of $\beta_{i+1}$ or $\beta_i$ is a $\imply$-Intro premise of $\beta_{i+1}$, and;
\item $\delta_k$ either is the conclusion of the derivation or the minor premise of a $\imply$-Elim.
\end{itemize}
\end{definition}

A normal derivation/proof in \mil is any derivation that does not have any formula occurrence that is simultaneously a major premise of a $\imply$-Elim and the conclusion of a $\imply$-Intro. A formula occurrence that is at the same time a conclusion of a $\imply$-Intro and a major premise of $\imply$-Elim is called a maximal formula. 
In \cite{Prawitz} proves the following theorem for the Natural Deduction for the full\footnote{The full propositional fragment is $\{\lor, \land, \imply, \neg, \bot\}$} propositional fragment of minimal logic.      

\begin{theorem}[Normalization]
  Let $\Pi$ be a derivation of $\alpha$ from $\Delta=\{\delta_1,\ldots, \delta_{n}\}$. There is a normal proof $\Pi^{\prime}$ of $\alpha$ from $\Delta^{\prime}\subseteq\Delta$.
\end{theorem}

In any normal derivation/proof, the format of a branch provides worth information on why huge proofs are redundant, as we will see in the next sections. Since no formula occurrence can be a major premise of $\imply$-Elim and conclusion of a $\imply$-Intro rule in a branch we have that the conclusion of a $\imply$-Intro can only be the minor premise of a $\imply$-Elim or it is not a premise of any rule application at all in the same branch. In this last case, it is the conclusion of the derivation or minor premise of a $\imply$-Elim rule. In any case, it is the last formula in the branch. Thus, for any branch, any conclusion of a $\imply$-Intro has to be a premise of a $\imply$-Intro. Hence, any branch in a normal derivation is divided into two parts (possibly empty). The E-part that starts the branch with the top-formula and every formula occurrence in it is the major premise of a $\imply$-Elim. There is a formula occurrence that is the conclusion of a $\imply$-Elim and premise of a $\imply$-Intro rule that is called minimal formula of the branch. The minimal formula starts the I-part of the branch, where every formula is the premise of a $\imply$-Intro, excepted the last formula of the branch. From the format of the branches, we can conclude that the sub-formula principle holds for normal proofs in Natural Deduction for \mil, in fact, for many extensions of it. A branch in $\Pi$ is said to be a principal branch if its last formula is the conclusion of $\Pi$. A secondary branch is a branch that is not principal. The primary branch is called a 0-branch. Any branch that the last formula is the minor premise of a rule in the E-part of a $n$-branch is a $n+1$-branch.

\begin{corollary}[Sub-formula principle]\label{coro:SubForProperty}
  Let $\Pi$ be a normal derivation of $\alpha$ from $\Delta=\{\delta_1,\ldots,\delta_m\}$. It is the case that for every formula occurrence $\beta$ in $\Pi$, $\beta$ is a sub-formula of either $\alpha$ or of some of $\delta_i$.
\end{corollary}

This corollary ensures that without loss of generality, any Natural Deduction proof of a \mil tautology has only sub-formulas of it occurring in it. In \cite{Exponential} we defined $EOL$-trees as an abstraction for normal proofs/derivations. They are forms of trees associated with derivations in Natural Deduction for \mil. The definition of $EOL$-tree facilitates the proof of the main result in \cite{Exponential}. With labelled trees, we can focus on the combinatorial aspects rather than the proof-theoretical. Unfortunately, in this article, we cannot do this. There are many important proof-theoretical details involved in this case. An abstraction of all of them would produce a very artificial concept.

To facilitate the presentation, we only handle normal proofs in expanded form.

\begin{definition} A normal proof/derivation is in expanded form, if and only if, all of its minimal formulas are atomic.
\end{definition}

We can consider, without loss of generality, that a formula in $\mil$ is a tautology if and only if there is a normal proof in expanded form that proves it. Of course, if it is tautology then it has a proof and so it has a normal proof by normalization. We use the following fact to obtain the expanded form from a normal proof.

\begin{proposition}\label{prop:NormalExpanded}
  Let $\Pi$ be a proof/derivation, in $\mil$, of $\alpha$ from $\Gamma=\{\gamma_1,\ldots,\gamma_k\}$. There is a proof in expanded form of $\alpha$ from $\Gamma$.
\end{proposition}

The proof of the above proposition is in  appendix~\ref{AppendixA}.

The following lemma~\ref{E-parts} shows that their respective top-formula uniquely defines the E-parts of any branch in a normal proof in expanded form. It uses the fact that if $\gamma_1\imply\gamma_2$ is the major premise of an application of $\imply-E$ then the conclusion is the right-hand side of this premise.

\begin{lemma}\label{E-parts}
  Let $\Pi$ be a normal proof in expanded form. Its respective top-formula uniquely determines each E-part's branch in $\Pi$.
\end{lemma}

\ifbool{formal}{
 \begin{proof} By induction on the degree of the top-formula\end{proof}}{
 \begin{proof}
  Every branch in $\Pi$ is of the form $\{\delta_{0},\ldots,q,\ldots,\delta_{k}\}$ where $\delta_{0}$ is a top-formula, and, $q=\delta_{j}$, an atomic formula,  is the minimal formula of the branch. By the definition of branch, definition~\ref{def:Branch} and the fact that $\Pi$ is a normal proof, for every $i=1,\ldots,j$, $\delta_{i-1}$ is major premise of an $\imply-E$ of a rule application having $\delta_{i}$ as conclusion. After $q=\delta_{j}$, the only possible rule applications are $\imply-I$, by the format of normal derivation's branches. So the sequence $\delta_{0},\ldots,\delta_{j}$ is maximal. Thus, from the top-formula $\delta_{0}$, we can obtain the whole sequence of formulas by picking up, recursively, the right-hand side of each of them. Finally, given the top-formula $\delta_{0}$, the whole sequence $\delta_{0},\ldots,\delta_{j}=q$ is determined. We can prove this by induction in the degree of the top-formula. 
 \end{proof}}

In the lemma~\ref{E-parts} above, each E-part's branch is uniquely induced from the top-formula. On the other direction, let $q$ be an atomic formula that is a minimal formula occurring in a branch $\Branch{b}$ in a normal proof in expanded form, $q$ does not determine the top-formula of $\vec{b}$ uniquely. For example, consider two branches:

\[
\{A\imply(B\imply (C\imply q)), (B\imply (C\imply q)),C\imply q, q,\ldots,\delta\}
\]
and
\[
\{B\imply (D\imply q),D\imply q,q,\ldots,\gamma\}
\}
\]
occurring in a normal proof $\Pi$. The minimal formula $q$ occurs in both. Given $q$, it is not possible to uniquely determine the top-formula of the branch to which it belongs to. However, if we observe with more attention, the $q$'s are not the same. The $q$of the first example is sub-formula of $A\imply(B\imply (C\imply q))$, while those in the second branch is sub-formula of $B\imply (D\imply q)$. If $\Pi$ is a normal proof of $\alpha$, the $q$'s are different occurrences in the syntax tree of $\alpha$. On the other hand, given a syntax tree $T_{\alpha}$ of $\alpha$ and an atomic formula $q$, we know that $q$ is a leaf in $T_{\alpha}$. There can be more than one leaf labelled with $q$, but given a specific leaf $q$ we can say from which top-formula it can be derived. The following lemma~\ref{lemma:Subformula} states this concerning any a normal proof $\Pi$
. Firstly we observe the following fact, which we state as a lemma without any proof.
\begin{lemma}
  Any formula in $\mil$ is of the form $(\alpha_0\imply (\alpha_1\imply \ldots (\alpha_k\imply q)\ldots)$, where
  $q$ is atomic.
\end{lemma}

Sometimes we use the notation $\left(\begin{array}{c}\alpha_0 \\ \ldots \\ \alpha_k\end{array}\right)\imply q$ to denote the  formula $(\alpha_0\imply (\alpha_1\imply \ldots (\alpha_k\imply q)\ldots)$ above.

As a consequence  of this lemma we have the following corollary~\ref{coro:Top-formulaq}

\begin{corollary}\label{coro:Top-formulaq}
  If $\Pi$ is a normal proof in expanded form and $q$ is the minimal formula of a branch $\Branch{b}$ then
  the top-formula of this branch is of the form $\left(\begin{array}{c}\alpha_0 \\ \ldots \\ \alpha_k\end{array}\right)\imply q$, for some $\alpha_i$, $i=1,k$.
\end{corollary}


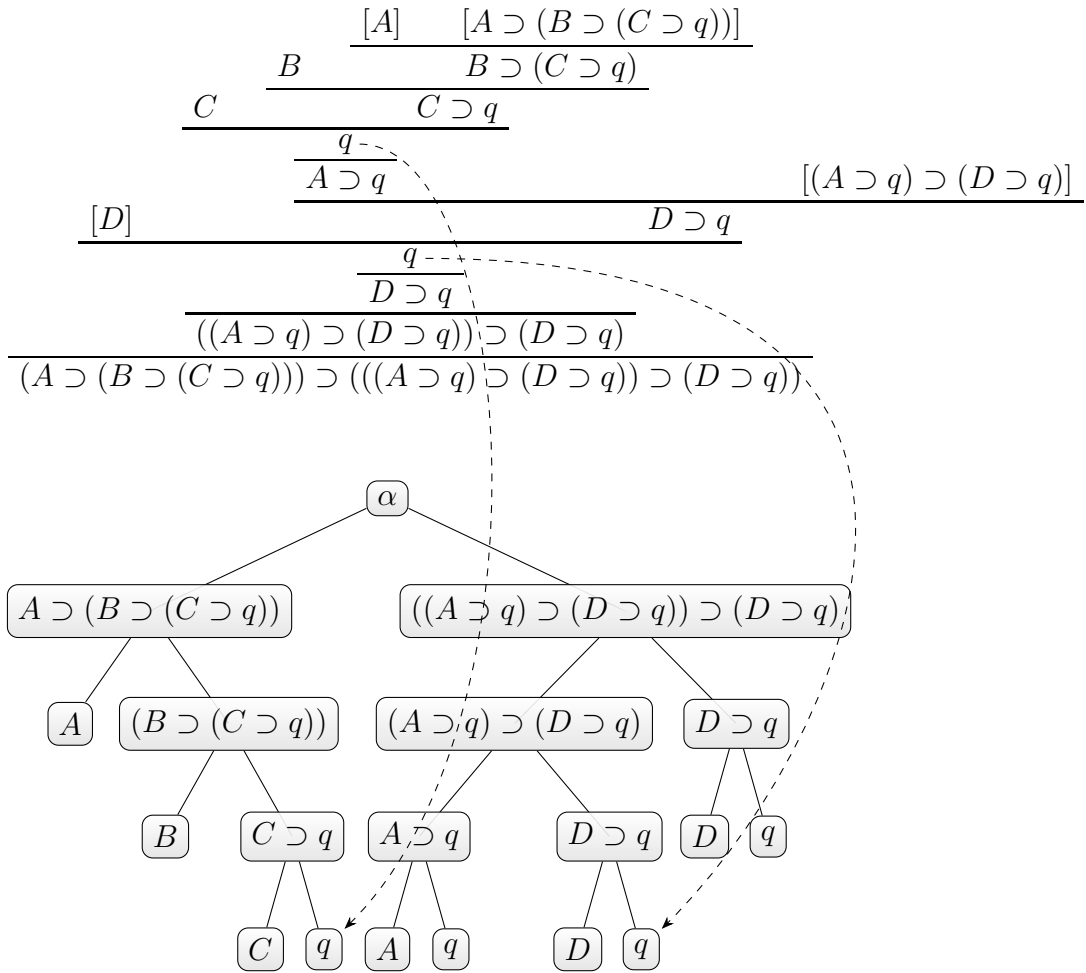
\begin{figure}
  \centering
 \centering
  \begin{tikzpicture}
    \node[anchor=south west] (proof) at (0,0) {%
        \AxiomC{$[D]$}
        \AxiomC{$C$}
        \AxiomC{$B$}
        \AxiomC{$[A]$}
        \AxiomC{$[A\imply (B\imply (C\imply q))]$}
        \BinaryInfC{$B\imply (C\imply q)$}
        \BinaryInfC{$C\imply q$}
        \BinaryInfC{$q$}
        \UnaryInfC{$A\imply q$}
        \AxiomC{$[(A\imply q)\imply (D\imply q)]$}
        \BinaryInfC{$D\imply q$}
        \BinaryInfC{$q$}
        \UnaryInfC{$D\imply q$}
        \UnaryInfC{$((A\imply q)\imply (D\imply q))\imply (D\imply q)$}
        \UnaryInfC{$(A\imply (B\imply (C\imply q)))\imply (((A\imply q)\imply (D\imply q))\imply (D\imply q))$}
    \DisplayProof};
  \begin{scope}[x={(proof.south east)},y={(proof.north west)}]
    \node (ANCHOR) at (.38,.36) {};
    \node (ANCHOR1) at (.32,.64) {};
  \end{scope}
  \begin{scope}[sibling distance=15em,
      every node/.style={shape=rectangle, rounded corners,
        draw, align=center,top color=white, bottom color=gray!20}]
  \node[below=of proof,xshift=-5em] {$\alpha$}
  child { [sibling distance=5em] node {$A\imply (B\imply (C\imply q))$}
    child { node {$A$} }
    child { [sibling distance=4em] node {$(B\imply (C\imply q))$}
      child {  node {$B$} }
      child { [sibling distance=2em] node {$C\imply q$}
        child { node {$C$} }
        child { node (b) {$q$} }
            }
          }
        }
    child { [sibling distance=7em] node {$((A\imply q)\imply (D\imply q))\imply (D\imply q)$}
      child { [sibling distance=6em] node {$(A\imply q)\imply (D\imply q)$}
        child { [sibling distance=2em] node {$A\imply q$}
          child { node {$A$} }
          child { node {$q$} }
              }
        child { [sibling distance=2em] node {$D\imply q$}
          child {  node {$D$} }
          child {  node (c) {$q$} }
              }
            }
      child { [sibling distance=2em]  node {$D\imply q$}
          child {  node {$D$} }
          child {  node {$q$} }
    } };
    \path (0,0) rectangle (1,-2);
    \draw[->,shorten >=1pt,>=Stealth,dashed] (ANCHOR)
    to [out=0, in=45,out looseness=.5, out min distance=20em] (c);
    \draw[->,shorten >=1pt,>=Stealth,dashed] (ANCHOR1)
    to [out=0, in=45,out looseness=.5, out min distance=-20em] (b);
  \end{scope}
  \end{tikzpicture}
  \caption{A mapped N.D. proof}
  \label{figure:MappedProofs}
\end{figure}

\begin{lemma}\label{lemma:Subformula}
  Let $\Pi$ be a normal proof of $\alpha$ in expanded form. Let $v$ be a leaf in $T_{\alpha}$, labeled with an atomic formula $q$. If $q$ is the minimal formula of a branch $\Branch{b}$ in $\Pi$, then there is at most one vertex $u$ in $T_{\alpha}$ that it is right-ancestral of $v$ and left child of a node $u$ labelled with the top-formula of $\Branch{b}$.
\end{lemma}

 \begin{proof} of lemma~\ref{lemma:Subformula}
  If $q$ is the minimal formula of $\Branch{b}$ then the top-formula of $\Branch{b}$ is of the form $\left(\begin{array}{c}\alpha_0 \\ \ldots \\ \alpha_k\end{array}\right)\imply q$ and the E-part of $\Branch{b}$ is of as follows:
  \[
  \left\{\left(\begin{array}{c}\alpha_0 \\ \ldots \\ \alpha_k\end{array}\right)\imply q,
  \left(\begin{array}{c}\alpha_1 \\ \ldots \\ \alpha_k\end{array}\right)\imply q,
  \left(\begin{array}{c}\alpha_2 \\ \ldots \\ \alpha_k\end{array}\right)\imply q,\ldots,
  \alpha_k \imply q, q\right\}
  \]
  By the form of this E-part's branch, we can see that all of its formulas, except the last one, are right-ancestral of $u$ in $T_{\alpha}$. Moreover $u$, labeled with $\left(\begin{array}{c}\alpha_0 \\ \ldots \\ \alpha_k\end{array}\right)\imply q$ has to be at the antecedent of a (sub) formula in $\alpha$, so it is a left-child of some vertex in $T_{\alpha}$. 
\end{proof}

 Using the above lemma~\ref{lemma:Subformula}, we can map each minimal formula, in a normal and expanded proof $\Pi$, employing a one-to-one correspondence to the respective top-formula occurrence of its branch. It is enough to use the vertexes of $T_{\alpha}$ for labelling the nodes of the proof-tree. Figure~\ref{figure:MappedProofs} illustrates the  necessity of a mapping from the proof into the syntax-tree of the proved formula. Note that the two positions of the atomic formula $q$ in the syntax tree uniquely indicates the top-formula in the E-part of the Natural Deduction proof/derivation to which it belongs. We can consider that the two $q$'s are in fact differents. The top-formula of each $q$ is the biggest one in the inverse path (upwards) following the reverse of the right child edge.  
 Definition~\ref{def:E-mapped-ND}, in the sequel, has this purpose. With the sake of a more precise presentation, we provide below the definition of a syntax tree of a formula.

\begin{definition}[Syntax tree of a formula]
  Let $\alpha$ be a \mil-formula. The syntax tree of $\alpha$ is the triple $\langle V,E_{left},E_{right},L\rangle$ where $V$ is a set, of vertexes, $E_{s}\subseteq V\times V$, $s=left,right$, the corresponding left and right edges, such that $\langle V,E_{left},E_{right}\rangle$ is an ordered full binary tree, and, $L$ is a bijective function from $V$ onto the subformulas of $\alpha$, such that:
  \begin{itemize}
  \item $L(r)=\alpha$, where $r\in V$ is the root of the tree $\langle V,E_{left},E_{right}\rangle$, and;
  \item For every formula $\varphi_1\imply\varphi_2\in Sub(\alpha)$, if $L(v)=\varphi_1\imply\varphi_2$,  $\langle v,v_1\rangle\in E_{left}$ and $\langle v,v_2\rangle\in E_{right}$ then $L(v_1)=\varphi_1$ and $L(v_2)=\varphi_2$.
  \end{itemize}
  \end{definition}

\begin{definition}[Partially mapped ND-proofs]\label{def:Mapped-ND}
  Let $\alpha$ be a \mil-formula and $T_{\alpha}=\langle V,E_{left},E_{right},L\rangle$ its syntax tree. Let $\Pi$ be a \mil-ND normal derivation of $\alpha$. A partially mapped ND-proof of $alpha$ is a structure $\langle \Pi,T_{\alpha},l\rangle$, where $l$ is a partial function from the formula occurrences in $\Pi$ to $V$, such that, the following conditions hold.
  \begin{itemize}
  \item If $\gamma$ is the minimal formula of a branch $\Branch{b}$ in $\Pi$ then if $l(\gamma)$ is defined then $L(l(\gamma))=\gamma$;
    \item If $\gamma$ is the minimal formula of a branch $\Branch{b}=\langle b_0,\ldots,b_j=\gamma,\ldots,b_k\rangle$ and $l(\gamma)$ is defined then either $l(b_{j-1})$ or $l(b_{j+1})$ are defined, and;
    \item If $\varphi_2$ is the conclusion of a $\imply$-Elim rule in $\Pi$, that has premises $\varphi_1\imply\varphi_2$ and $\varphi_1$,  and $l(\varphi_2)=v_2$ then there are $v$ and $v_1$, such that $\langle v, v_2\rangle\in E_{right}$, $\langle v,v_1\rangle\in E_{left}$, $l(\varphi_1)=v_1$ and $l(\varphi_1\imply\varphi_2)=v$;
          \item If $\varphi_1\imply\varphi_2$ is the conclusion of a $\imply$-Intro rule in $\Pi$, that has premise $\varphi_2$ and $l(\varphi_1\imply\varphi_2)=v$ then there is $v^{\prime}\in V$, $\langle v,v^{\prime}\rangle\in E_{right}$ and $l(v^{\prime})=\varphi_2$. 
      \end{itemize}
  \end{definition}

\begin{definition}[E-mapped Natural Deduction Normal Expanded proofs]\label{def:E-mapped-ND}
  Let $\alpha$ be a \mil-formula, $T_{\alpha}=\langle V,E_{left},E_{right},L\rangle$ be the syntax tree of $\alpha$ and $\Pi$ a normal and expanded proof of $\alpha$. The triple $\langle \Pi,T_{\alpha},l\rangle$ is an E-mapped Natural Deduction proof, if and only if, $l$ is defined on all formula occurrences that take part in the E-parts of branches in $\Pi$, including the minimal formulas. Moreover the following condition must hold:
  \begin{itemize}
  \item For every branch $\Branch{b}$, if $q$ is the minimal formula of $\Branch{b}$, $l(q)=v\in V$ and $\beta$ the top-formula (occurrence) of $\Branch{b}$ then $l(\beta)=u$, where $u$ is the right-ancestral of $v$ that is left-child of some $w\in V$, as stated by lemma~\ref{lemma:Subformula}.
  \end{itemize}
  \end{definition}

Since lemma~\ref{lemma:Subformula} holds then the above definition is well-defined. Moreover, we have the following proposition. We use the acronym {\bf EmND} to refer to an E-mapped Natural Deduction Normal Expanded proof. In the following, we consider a branch as a sequence of formula occurrences numbered from top-formula down to the last formula of the branch. 

\begin{proposition}\label{prop:one2one}
  Let $\langle \Pi,T_{\alpha},l\rangle$ be a {\bf EmND} of $\alpha$. We have that to each branch $\Branch{b}=\langle \beta_0,\ldots,\beta_k,$ in $\Pi$ that has the  minimal formula occurrence $q=\beta_{j}$, such that $l(\beta_{j})=u\in V$,  there exists one and only one path $p=\langle u_0,\ldots,u_j\rangle$ in $T_{\alpha}$, with $u_j=u$, and $u_0$ as stated in lemma~\ref{lemma:Subformula}, such that, $l(\beta_{i})=u_i$, $i=0,\ldots,j$. 
\end{proposition}

In the above proposition we can also see that $\langle \beta_0,\ldots,\beta_j=q\rangle$ is the E-part of $\Branch{b}$.

\ifbool{formal}{
 \begin{proof} This a immediate consequence of corollary~\ref{coro:Top-formulaq} and lemma~\ref{lemma:Subformula} \end{proof}}
{\begin{proof}
    From corollary~\ref{coro:Top-formulaq} applied to vertex $u$, such that $l(u)=q$, we have a unique top-formula that labels the vertex $v$ of $T_{\alpha}$. Moreover,  $v$ is the first right-ancestral of $u$ that is left-child of some vertex in $T_{\alpha}$, hence, the sequence $\{v=u_0,u_1,\ldots,u_k,u\}$ is the $p$, the path of vertexes such that $\{l(u_0),l(u_1),\ldots,l(u_k),l(u)=q\}$ is the whole E-part of the branch from which the occurrence of $q$ is the minimal formula and $l(v)=l(u_0)$ is top-formula. 
\end{proof}}

The importance of proposition~\ref{prop:one2one} is that it states that any given E-part $\langle \beta_0,\ldots,\beta_j\rangle$ of a branch in an {\bf EmND} is an instance of at most one path $p=\langle u_0,\ldots,u_j\rangle$ in $T_{\alpha}$, such that $L(u_i)=\beta_i$, $i=0,\ldots,j$. Moreover, this path $p$ is as stated in lemma~\ref{lemma:Subformula}.  Given a {\bf EmND} $\Pi$,
for  each  E-part,  in an {\bf EmND},  exists a path of the form of lemma~\ref{lemma:Subformula}, in the syntax tree of the conclusion of the{\bf EmND}. The number of such paths in the syntax tree is upper-bounded by its size, then the number of different E-parts types in any {\bf EmND} is at most of the size of the conclusion of this {\bf EmND}. We have the following lemma:

\begin{lemma}[Linear upper-bound on types of E-parts]\label{lemma:Upper-boundsE-parts}
  Let $\Pi$ be an {\bf EmND} proving the \mil-formula $\alpha$. The number of different types of E-parts occurring in this {\bf EmND} is at most the size of the $T_{\alpha}$.
\end{lemma}

This lemma is very useful in the next section. Moreover, we remark that if we label the nodes of a Natural Deduction proof-tree with the nodes (not the labels) of the syntax tree of the conclusion of the proof-tree, we will have the same effect on counting different types of E-parts that  lemma~\ref{lemma:Upper-boundsE-parts} above reports. 
  
\section{Counting repeated patterns in polynomially lower-bounded proofs}\label{sec:Lower-bounds}

In \cite{Studia2019}, theorem 4, page 57,  we show that we can prove any \mil tautology using a normal proof that has height upper-bounded by a linear function on the size of the formula. If we consider normal proofs in expanded form, the upper-bounded on the height of the proof is still linear on the size of the conclusion too. This last statement is proved in appendix, proposition~\ref{appendix:linear-upper-bound-on-expanded}. This section proves some auxiliary lemmas that are useful to prove the main results in the next section~\ref{sec:Main}.

Because of the linear speedup theorem, see \cite{LinSpeedUp-Space} page 63-64, Theorem 3.10, w.l.g., we consider that a linear height bounded proof of $\alpha$ is a proof which height is upper-bounded by the length of $\alpha$. In fact, in this article, because of counting details we consider that the upper-bounded is the size of the syntax tree $\|T_{\alpha}\|$. Since $\|T_{\alpha}\|=\|\alpha\|$, the definition is equivalent. 

  \begin{lemma}[Spreding Branchs Repetitions]\label{Zero} Let $\langle \Pi,T_{\alpha},l\rangle$ be a linearly height bounded {\bf EmND} proof of $\alpha$, $0<p\in \mathbb{N}$ and $m=\|\alpha\|$. If there is a branch $\Branch{b}$ that has more than $m^p$ instances occurring in $\Pi$ then there is a level $\mu$, such that, at least $m^{p-1}$ instances of $\Branch{b}$ have the minimal formula $q_{\Branch{b}}$ of $\Branch{b}$ occurring in level $\mu$.
    \end{lemma}

  \begin{lemma}[Branchs and sub-derivations]\label{Um} Let $\Pi$ be a proof of $\alpha$, and $\Branch{b}$ a branch in $\Pi$ under the same conditions of the lemma~\ref{Zero} above. Then there is a (sub) derivation $\Pi_{\Branch{b}}$ of $\Pi$, such that, $\Pi_{\Branch{b}}$ has at least $m^{p-1}$ instances occurring in $\Pi$.
  \end{lemma}

\begin{proof}~of lemma~\ref{Zero}: Since $\Pi$ has its height bounded by $\|\alpha\|=m$, there are at most $m$ levels in $\Pi$. In order to  accommodate more than $m^p$ branches there must be a level $1\leq n\leq m$ that has at least $m^{p-1}$ branches with their respective conclusion occurring in this level. Thus, considering $k$ the length of the I-part of $\Branch{b}$, the minimal formula $q$ of $\Branch{b}$ occurs in level $\mu=n-k$ at least $m^{p-1}$ times.    
  
\end{proof}

To prove the lemma~\ref{Um}, we need the definition of the reverse rank of a branch instance in a proof $\Pi$, namely, $rr_{\Pi}(\Branch{b})$.

\begin{definition}\label{rr}
Given a Natural Deduction derivation (proof) $\Pi$ in $\mathcal{M}_{\imply}$, we define the reverse rank of a branch $\Branch{b}$ in $\Pi$, $rr_{\Pi}(\Branch{b})$ as following:
\begin{enumerate}
\item\label{rrZero} $rr_{\Pi}(\Branch{b})=0$, iff, $\Branch{b}$ has no conclusion of a $\imply$-E rule application;
\item\label{rrInd} Let $\Branch{b_1},\ldots,\Branch{b_k}$ be the branches instances in $\Pi$ with last formulas coorurrences $\alpha_1,\ldots,\alpha_k$, respectively, such that, $b_0,\alpha_1,\ldots,\alpha_k$ is the E-part of $\Branch{b}$ with $b_0$ the top-formula of this branch. $rr_{\Pi}(\Branch{b})=max(\{rr(\Branch{b_1}),\ldots,rr_{\Pi}(\Branch{b_k})\})+1$.
\end{enumerate}
\end{definition}

We use $rr(\Branch{b})$ whenever $\Pi$ in $rr_{\Pi}(\Branch{b})$ can be infered from the scope. 

\begin{proof}~ of lemma~\ref{Um}:
We reinforce that the conditions in lemma~\ref{Um} are the same in lemma~\ref{Zero}.
Thus, let $\Pi$ be a proof of $\alpha$, $m=\|\mathcal{T}(\alpha)\|$, $p>0$, $p\in\mathbb{N}$, such that, there is a branch $\Branch{b}$ that has more than $m^p$ instances occurring in $\Pi$. By applying lemma~\ref{Zero} we obtain the existence of a level $\mu$ where at least $m^{p-1}$ instances of the minimal formula of $q_{\Branch{b}}$ occur in this level $\mu$ in $\Pi$. We now prove by induction on $rr(\Branch{b})$ the lemma.
\begin{description}
\item[Base] In this case $rr(\Branch{b})=0$, so $\Branch{b}$ has only I-part, having no $\imply$-E rule conclusions. Thus, $q_{\Branch{b}}$ is a top-formula, and hence, the branch $\Branch{b}$ itself is a valid (sub)derivation in $\Branch{b}$. It has $m^{p-1}$ instances occurring in level $\mu$ in $\Pi$. 
\item[Inductive] Consider $\Branch{b}$, such that $rr(\Branch{b})>0$. 
So, $\Branch{b}$ has a non-empty E-part. Let $b_0,\alpha_1,\ldots,\alpha_k$ be this E-part, with $\Branch{b_1},\ldots,\Branch{b_k}$ the list of all branches secondary to $\Branch{b}$. By definition of branch, we have that $rr(\Branch{b_{i}})<rr(\Branch{b})$, for $i=1,k$,  so by inductive hypothesis  there are $\Pi_{b_{i}}$ (sub)derivations of $\Pi$, $i=1,k$. We remember that $b_0$ is the top-formula of $\vec{b}$. There is one occurrence of $\Pi_{b_{i}}$ to each $\Branch{b}$ instance in $\Pi$. This is a consequence of \ref{lemma:Upper-boundsE-parts}. Thus, joining all of these (sub)derivations in a whole (sub)derivation $\Pi_{\Branch{b}}$ shows us the existence of an equal number of instances of it as sub derivation of $\Pi$. Thus, summing up, there are at least $m^{p-1}$ instances of this joined subderivations having the corresponding $\Branch{b}$ instances as its main branch in $\Pi$. 
\end{description}

\end{proof}

\section{Redundancy in huge $\mathcal{M}_{\imply}$ mapped derivations}\label{sec:Main}


  Let $\Phi$ be the set of mapped linearly height-bounded ND \mil proofs. We use the notation $c(\Pi)$ to denote the formula that is the conclusion of $\Pi$. Note that $\Phi$ can be considered as a predicate $\Phi(x)$ that is true if and only if $x$ is assigned to a mapped linearly height-bounded ND proof. 
\[
S_{\Phi}=\{\mbox{$\Pi\in \Phi$}:\mbox{$\forall p\in\mathbb{N}, p>0, \exists n_0,\forall n>n_0$, $\|T_{c(\Pi)}\|=n$ and $\|\Pi|>n^p$)}\}
\]


In appendix~\ref{AppendixCC} it is discussed how hard is a set as the above one to be computationally verified. In fact, the following proposition is one of the reasons to work with sets as the above-defined one. Informally $S_{\Phi}$ contains all huge or hard linearly height upper-bounded proofs in \mil. Of particular interest is the following set. Let $Taut_{\mil}$ be the set of all ND mapped proofs of \mil tautologies. The following set:

\begin{definition} Let $\Theta$ be the following set:
  \[
  \Theta_{\mil} = \{\mbox{$\Pi\in Taut_{\mil}$}:\mbox{$\forall p\in\mathbb{N}, p>0, \exists n_0,\forall n>n_0$, $\|T_{c(\Pi)}\|=n$ and $\|\Pi|>n^p$)}\}
  \]
\end{definition}

$\Theta_{\mil}$ is the set of super-polynomially sized ND mapped proofs.

The primary purpose of this article is to show that every $\Pi\in S_{\Phi}$ is redundant.  I.e.,  there is at least one sub-proof $\Pi_{s}$  of $\Pi$ that repeats as many times as it is the size of $\Pi$. A consequence of the following theorem~\ref{main}. Remember that all proofs in $S_{\Phi}$ are linearly heigh-bounded.

\begin{theorem}\label{main} For all $p\in\mathbb{N}$, $p>3$, and for all $\Pi\in S_{\Phi}$, such that, $\|\mathcal{T}(c(\Pi))\|=m$ and $\|\Pi\|>m^p$,  then there is a sub-derivation $\Pi_{s}$ of $\Pi$ and a level $\mu$ in $\Pi$, such that, $\Pi_{s}$ has at least $m^{p-3}$ instances occurring in the level $\mu$ in $\Pi$.
\end{theorem}


\begin{proof}~ of theorem~\ref{main}
Consider $p\in\mathbb{N}$, $p>2$, then there is $m_0\in\mathbb{N}$, such that for all $m>m_0$, and $\Pi\in S_{\Phi}$,  $\|T_{c(\Pi)}\|=m$ and $\|\Pi\|>m^p$. Let $T=T_{c(\Pi)}$ and $U_{F,\lambda}(\Pi)=\{v:\mbox{$v$ is in level $\lambda$ and $l(v)=F$}\}$ then: 
\[
\|\Pi\|=\sum_{\lambda=1,h(T)}\sum_{F\in V(T)} \|U_{F,\lambda}(T)\| > m^p
\]
Thus, there are $0<\mu\leq m-1$ and $F\in V(T)$, such that $\|U_{F,\mu}|>m^{p-2}$. Considering that $U_{F,\mu}=U_{F,\mu}(\Pi)$ and analogously for $Top$, $Uno$ and $Duo$, we have that:
\[
U_{F,\mu}= Top_{F,\mu}\oplus Uno_{F,\mu}\oplus Duo_{F,\mu}
\]
, where $Top_{F,\mu}$ is the subset of $U_{F,\mu}$ of $F$-occurrences in $\Pi$ in level $\mu$ as top-formulas, $Uno_{F,\mu}$ is the subset of $U_{F,\mu}$ of $F-occurrences$ in $\Pi$, in level $\mu$, as conclusions of $\imply$-I introductions, and, $Duo_{F,\mu}$ are the $F$-occurrences in $\Pi$, in level $\mu$, as conclusions of $\imply$-E rules. Thus, we have at least one of the following alternatives that we analyze in the sequel. 
\begin{description}
\item[Top-Formulas] $\|Top_{F,\mu}|>m^{p-2}$ and in this case we reach the conclusion of the theorem, for $F$, itself, is a sub-derivation of $\Pi$ that has more than $m^{p-2}$  instances occurring in $\Pi$.
\item[$\imply$-I] $\|Uno_{F,\mu}|>m^{p-2}$ and in this case, there is a sequence of length $k$, $k>0$, of $\imply$-I rules for each branch $\Branch{b}$ with  $F\in Uno_{F,\mu}$. Thus, the minimal formula of each branch $\Branch{b}$ occurs in level $\mu-k$. Thus we have that the branch $\Branch{b}$ occurs $m^{p-2}$ times in $\Pi$. Finally, by applying lemma~\ref{Um}, we conclude that there is a sub-derivation $\Pi_{s}$ of $\Pi$ that has at least $m^{p-3}$ instances in $\Pi$.
\item[$\imply$-E]  Since $\Pi$ is a mapped derivation, then, all $F$'s, as instances of conclusions of $\imply$-E rules, have the same major premise, that is the same instance of a formula $F^{\prime}\imply F\in T_{c(\Pi)}$. Thus, there is a unique branch $\Branch{b}$ that contains both $F$ and $F^{\prime}\imply F$ as consecutive formulas. Thus, this branch has the same minimal formula instance, as well as the same top-formula and I-part, and, it has $m^{p-2}$ instances occurring in $\Pi$. Again by an application of lemma~\ref{Um}, we obtain the conclusion of the theorem.
\end{description}
\end{proof}

  From theorem~\ref{main} we can roughly state the corollary~\ref{roughly}.

\begin{corollary}\label{roughly} In every family of super-polynomial and linearly height upper-bounded proofs all of them are super polynomial redundant.
\end{corollary}

If we inspect the above theorem~\ref{main} proof, we can see that we can replace the condition on {\bf linearly height upper-bounded} mapped proofs by any fixed polynomial height upper-bounded class of mapped proofs. Thus,  if we  define $\Phi(q)$ as the class of $n^q$ height upper-bounded mapped, normal and expanded, proofs, and
\[
S_{\Phi}^q=\{\mbox{$\Pi\in \Phi(q)$}:\mbox{$\forall p\in\mathbb{N}, p>0, \exists n_0,\forall n>n_0$, $\|T_{c(\Pi)}\|=n$ and $\|\Pi|>n^p$)}\}
\]then theorem~\ref{main} becomes:

\begin{theorem}\label{generalmain}
   For all $p\in\mathbb{N}$, $p>q+2$, and for all $\Pi\in S_{\Phi}^q$, such that, $\|\mathcal{T}(c(\Pi))\|=m$ and $\|\Pi\|>m^p$,  then there is a sub-derivation $\Pi_{s}$ of $\Pi$ and a level $\mu$ in $\Pi$, such that, $\Pi_{s}$ has at least $m^{p-(q+3)}$ instances occurring in the level $\mu$ in $\Pi$.
\end{theorem}

Finally, we can conclude that the corollary~\ref{roughlyII} holds concerning $\Phi(q)$ too.

\begin{corollary}\label{roughlyII} In every family of super-polynomial and fixed polynomial height upper-bounded proofs all of them are super polynomial redundant.
\end{corollary}

\section{Conclusion}\label{sec:Conclusion}

This article precisely defines huge proofs within the system of Natural Deduction for the Minimal implicational propositional logic \mil. This is what we call unlimited family of super-polynomial proofs. We consider huge families of expanded normal form mapped proofs, a device to explicitly help to count the E-parts of a normal proof in an adequate way. Finally, we show that for almost all members of the a super-polynomial family there at least one sub-proof or derivation of each of them that is repeated super-polynomially many times. This last property we call super-polynomial redundancy. Summing up, in this article we show that huge proofs are highly redundant. The main application of the result we show here is in compressing ND proofs into DAGs (Directed Acyclic Graphs) that can be viewed as polynomial certificates for provability in \mil. The compression is a refinement of the horizontal compression presented in \cite{Studia2019} and \cite{BSLO2020}. In the later,\cite{BSLO2020}, the path certificates are added to the DAGs certificates  to obtain polynomial certificates for provability in \mil, proving that NP=PSPACE. The refinement of the cited  compression that we develop here in this article can be used together with the linear height normal proofs for non-hamiltonicity, presented in the appendix of \cite{Exponential},  to provide a simpler proof of NP=CoNP that does not need \cite{Hudelmaier}.

\appendix

\section{Some useful proof-theoretical results}\label{AppendixA}

\begin{proposition}\label{appendix:greedy-imply-intro}
  Let $\Pi$ be a proof of $\alpha$ in \mil. Then there is a proof of $\alpha$ in \mil, where every  $\imply$-Intro applications are greedy.
\end{proposition}

\begin{proof} Since $\Pi$ is a proof, every top-formula is discharged by an application of $\imply$-Intro rule. Fix any top-formula that it is not discharged by the first $\imply$-Intro application top-down. This top-formula must be discharged in a subsequent $\imply$-Intro downnwards. If we apply greedy $\imply$-Intro applications, the first application is greedy and discharges the formula, while all subsequent $\imply$-Intro are $\imply$-Intro vacuous applications. This works for each greedy $\imply$-Intro applications.
  \end{proof}

We have the following corollary.

\begin{corollary}\label{coro:NDWithSetsOdDependencies}
    For every \mil tautology $\alpha$ there is a proof in $\mil^{\rightarrow}$ of $\alpha$
    \end{corollary}  

\begin{proposition}\label{appendix:linear-upper-bound-on-expanded}
  Let $\Pi$ be a proof/derivation, in $\mil$, of $\alpha$ from $\Gamma=\{\gamma_1,\ldots,\gamma_k\}$. There is a proof in expanded form of $\alpha$ from $\Gamma$.
\end{proposition}

\begin{proof} Proof of proposition~\ref{appendix:linear-upper-bound-on-expanded} and proposition~\ref{prop:NormalExpanded}.  
  If $\varphi_1\imply \varphi_2$ is a minimal formula in some branch of $\Pi$. We replace $\varphi_1\imply\varphi_2$ by:
  \begin{prooftree}
    \AxiomC{$[\varphi_1]$}
    \AxiomC{$\varphi_1\imply\varphi_2$}
    \BinaryInfC{$\varphi_2$}
    \UnaryInfC{$\varphi_1\imply\varphi_2$}
  \end{prooftree}
  Proceed  to the replacing, now about $\varphi_2$, until it is atomic.

  \end{proof}

\section{Super-polynomially sized propositional proofs/derivations}\label{AppendixCC}

This section is a variation of the corresponding section in \cite{Exponential}, where we define exponentially sized labelled trees as a counterpart of exponentially sized N.D. proofs/derivations.  

 Concerning the computational complexity of propositional proofs, we consider the size of a proof as the number of symbol occurrences used to write it, i.e., the length of the linearized proof-tree.  If we put all the symbol occurrences used to write a Natural Deduction derivation $\Pi$ side by side in a long string then the size of the derivation, denoted by $\size{\Pi}$, is the length of this string. The function $\size{\;}:Strings\longrightarrow \mathbb{N}$, the size-of-string function, denotes the mapping of strings to their corresponding sizes\footnote{Some authors use the term {\em lenght} instead of {\em size}}. For derivations $\Pi$ of $\alpha$ from $\Delta=\{\delta_1,\ldots, \delta_n\}$ we estimate the complexity of the derivation by means of a function of $\size{\alpha}+\sum_{i=1,n}\size{\delta_i}$ into the size of the derivation itself. Given $\Delta$ and $\alpha$, such that, $\Delta\models\alpha$, we know that there are infinitely many derivations $Pi$ of $\alpha$ from $\Delta$, even for normal derivations, there are formulas that have infinitely many normal proofs. Thus, an adequate estimation of the complexity of a tautology is to know how big it is the smallest proof of it when compared to the size of the tautology itself. This gives rise to a function $CC:\mathbb{N}\rightarrow\mathbb{N}$, as follows:

 \[
 CC(n)=min_{\alpha\in\mathcal{S}(n)}\{\size{\Pi}:\mbox{$\Pi$ is a proof of $\alpha$}\}
 \]
where $\mathcal{S}(n)$ is the set of all tautologies that have length $n$. The minimum of an empty set of formulas is 0. The complexity of recognizing provable formulas (tautologies) is no better than the lower-bound function that provides the size of the smallest Natural Deduction proof among the proof of all formulas of the same length, the $CC$ function above defined. Note the above function works on any logic that has finite proofs in a system like Natural Deduction.  We describe below another way of estimating the computational complexity of provable formulas by using a set of proofs. 

A set $\mathcal{S}$ of Natural Deduction proof-trees is unlimited, if and only if, for every $n>0$ there is $\Pi\in \mathcal{S}$, such that, $\size{\Pi}>n$. We remember that, in the following definition, $c(\Pi)$ denotes the formula that is the conclusion of $\Pi$. 


\begin{definition}\label{SP:local}
  An unlimited set $\mathcal{S}$ of ND proof-trees is super-polynomially big or $\mathcal{S}\mathcal{P}$ for short, or simply {\bf huge} iff for every $p\in\mathbb{N}$, $p\ge 1$, there are $n_0\in\mathbb{N}$ and $c\in\mathbb{R}$, $c>0$, such that, for every $n>n_0$ and for every $\Pi\in\mathcal{S}$, if $\size{c(\Pi)}= n$ then $\size{\Pi}\ge c\times n^p$.
\end{definition}

We use $\size{A}$ to denote the length of $A$. 

The following definitions and facts justify the primary purpose of the above definition.  


We remind the reader that the size of the alphabet used to write the strings is at least 2. Unary strings cannot be consistently used in computational complexity estimations, since its use trivializes\footnote{If there is a NP-complete Formal Language $L\subseteq\Sigma^{\star}$, where $\Sigma$ is a singleton, then $NP=P$, see for example \cite{Creszenci} (theorem 5.7, page 87)} the conjecture $NP=P$. We use to call an alphabet reasonable whenever it has at least two symbols. 

\begin{definition}
  A function  $f:\mathbb{N}\longrightarrow\mathbb{N}$ is super-polynomial if and only for any polynomial $n^p$, $p>1$, $p\in \mathbb{N}$, $f(n)$ is bigger than $n^p$ for almost all $n\in\mathbb{N}$. Formaly, for any $p\in\mathbb{N}$, $p>1$, there are $n_0\in\mathbb{N}$ and $c\in\mathbb{R}$, $p\geq 1$, $c>0$, such that, $\forall n>n_0$, $f(n)\ge c\times n^p$.
\end{definition}

It is worth noting that the constant $c$ in the definition above represents the scale invariance typical whenever we compare computational complexities.
Technically, the above definition says that a function is super-polynomial whenever it is lower-bounded by any polynomial.  

Consider a property $\Phi(x)$ on N.D. proof-trees. This property is used to select, from a set $\mathcal{S}$ of proof-trees, all the proof-trees satisfying it. This defines a subset $\{\Pi\in \mathcal{S}:\Phi(\Pi)\}$ of $\mathcal{S}$. As an example we can set a particular $\Phi_{\Gamma,\alpha}(x)$, where $\Gamma$ is a set of \mil-formulas and $\alpha$ is a \mil-formula, to be true only on N.D. proof-trees $\Pi$, such that $\Gamma$ is the set of open assumptions and $c(\Pi)=\alpha$. Thus, given a set $\mathcal{S}$ of proof-trees, the set $\{\Pi\in \mathcal{S}:\Phi_{\Gamma,\alpha}(\Phi)\}$ is the subset of all ND proof-trees from $\mathcal{S}$ that are derivations of $\alpha$ from $\Gamma$. We further refine this to get the set of all minimal trees (derivations) of $\alpha$ from $\Gamma$. For example
\[
Min_{\mathcal{S}}(\Gamma,\alpha)=\{\Pi\in\mathcal{S}:\mbox{$\Phi_{\Gamma,\alpha}(\Pi)\;\land\;\forall \Pi^{\prime}(\Phi_{\Gamma,\alpha}(\Pi^{\prime})\;\rightarrow\; \size{\Pi}\leq\size{\Pi^{\prime}})$}\}
\]
The above set is the set of the smallest  N.D. proof-trees that satisfy $\Phi_{\Gamma,\alpha}(x)$. They are the set of all smallest derivations of $\alpha$ from $\Gamma$ in $\mil$. In the general case, where the predicate $\Phi(x)$ is arbitrary, we denote the set above by $Min_{\mathcal{S}}(\Phi)$, that is:
\[
Min_{\mathcal{S}}(\Phi)=\{\Pi\in\mathcal{S}:\mbox{$\Phi(\Pi)\;\land\;\forall \Pi^{\prime}(\Phi(\Pi^{\prime})\;\rightarrow\; \size{\Pi}\leq\size{\Pi^{\prime}})$}\}
\]

\begin{definition}\label{def:FuncaoF}
  Let $\mathcal{S}$ be an unlimited set of N.D. proof-trees. Let $\Phi(x)$ represent a property on N.D. proof-trees of $\mathcal{S}$ and let $\Phi_{\mathcal{S},m}(x)$ be defined as $(x\in\mathcal{S}\;\land\;\Phi(x)\;\land\;\size{c(x)}\leq m)$ with $0<m\in\mathbb{N}$. We define the function $F_{\mathcal{S},\Phi}:\mathbb{N}\longrightarrow\mathbb{N}$ that associates do each natural number $m$ the least N.D. proof-tree satisfying $\Phi_{\mathcal{S},m}(x)$.
  \[
  F_{\mathcal{S},\Phi}(m)=\left\{\begin{array}{ll} 0 & \mbox{if $m=0$} \\
                                      \size{Min_{\mathcal{S}}(\Phi_{\mathcal{S},m})} & \mbox{if $m>0$}               
  \end{array}\right.
  \]
\end{definition}

We point out that depending on $\Phi$, the above function $F_{\mathcal{S},\Phi}(m)$ can be quite uninteresting. For example, if $\Phi$ is satisfiable by every ND proof-tree in $\mathcal{S}$ then  $F_{\mathcal{S},\Phi}(m)=1$, for every $m>0$. Any ND proof-tree $\Pi$ with only one node,  such that, it is $c(\Pi)$, is a smallest\footnote{We do not take the null tree in this work since it represents no meaningful representation of data in our case} N.D. proof-tree that satisfies $\Phi$. On the other hand, we can have $\Phi_{\mathcal{S},m}(\Pi)$ true only when $\Pi$ is a proof-tree that represents a proof of a \mil tautology $\alpha$, $m=\size{\alpha}$.

The following proposition points out an alternative and somtimes more adequate definition for a family of super-polynomially sized proof-trees as already previously mentioned. Observe that if $\mathcal{A}$ is the set of all proof-trees and $\Phi(x)$ is a property defining a subset $\mathcal{S}$ of $\mathcal{A}$ and $\Phi_{\mathcal{S},m}$ is defined as in definition~\ref{def:FuncaoF} then $\mathcal{S}=\Phi(\mathcal{A})=\bigcup_{m\in\mathbb{N}}\Phi_{\mathcal{A},m}(\mathcal{A})$. The reader should note that we use $\Phi(\mathcal{A})$ as an abbreviation of $\{\Pi:\mbox{$\Phi(\Pi)\land \Pi\in\mathcal{A}$}\}$. Observing what is discussed in the last paragraphs, we have the following proposition.

\begin{proposition}\label{prop:FuncaoF}
  Let $\mathcal{S}\subset\mathcal{A}$ be an unlimited set of proof-trees. Let $\Phi(x)$ be the defining property of $\mathcal{S}$. We have then that  $\mathcal{S}$ is \SP\; if and only if $F_{\mathcal{A},\Phi}$ is a super-polynomial  function from $\mathbb{N}$ in $\mathbb{N}$.
  \end{proposition}

In this article, we are interested in families of super-polynomial proofs. In particular, a family $\mathcal{S}$ of super-polynomial proofs in \mil is an unlimited set of proof-trees, satisfying definition~\ref{SP:local}. The proposition above provides the soundness of definition~\ref{SP:local} concerning the lower-bound for a set of computational objects (proofs). In section~\ref{sec:Main} we show that any set of \SP\ family of proofs is intrinsically redundant, i.e., almost all of its elements have super-polynomially many repetitions of a  pattern.


\end{document}